\documentclass[12pt]{article}

\usepackage[a4paper,margin=1in]{geometry}
\usepackage[numbers]{natbib}
\usepackage{lipsum}
\usepackage{balance} % for balancing columns on the final page
\usepackage{blindtext}
\usepackage{hyperref}
\usepackage{xcolor}
\usepackage{amsfonts}
\usepackage{multirow}
\usepackage{makecell}
\usepackage{tabu}
\usepackage{graphicx}
\usepackage{enumitem}
\usepackage{algorithm}
\usepackage{algpseudocode}
\usepackage[english]{babel}
\usepackage{amsthm}
\usepackage{amsmath}
\usepackage{cleveref}
\usepackage{array}
\usepackage{booktabs}
\usepackage{soul}
\usepackage{authblk}
\usepackage{amsmath}
\usepackage[bibliography=common]{apxproof}

\newcommand{\allocs}{\mathcal{A}}

\newcommand{\BibTeX}{\rm B\kern-.05em{\sc i\kern-.025em b}\kern-.08em\TeX}
\newtheorem{Theorem}{Theorem}[section]

\newtheorem{Proposition}[Theorem]{Proposition}

\newtheorem{Corollary}[Theorem]{Corollary}

\theoremstyle{definition}
\newtheorem{Definition}[Theorem]{Definition}
\newtheorem{Example}[Theorem]{Example}

\newtheorem{Open}[Theorem]{Open Question}

\newcommand{\no}[1]{\textcolor{red}{#1}}
\newcommand{\noga}[1]{\no{(Noga says: #1)}}
\newcommand{\fullversion}[1]{}

\makeatletter
\newcounter{phase}[algorithm]
\newlength{\phaserulewidth}

\makeatother

\makeatletter

\title{The Min Max Average Cycle Weight Problem}
\author[1]{Noga Klein Elmalem}
\author[2]{Rica Gonen}
\author[3]{Erel Segal-Halevi}

\affil[1]{The Open University of Israel\\noga486@gmail.com}
\affil[2]{The Open University of Israel\\ricagonen@gmail.com}
\affil[3]{Ariel University, Israel\\erelsgl@gmail.com}

\AtBeginEnvironment{proof}{\footnotesize}

\begin{document}
\maketitle

\begin{abstract}
When an old apartment building is demolished and rebuilt, how can we fairly redistribute the new apartments to minimize envy among residents? We reduce this question to a combinatorial optimization problem called the \emph{Min Max Average Cycle Weight} problem. In that problem we seek to assign objects to agents in a way that minimizes the maximum average weight of directed cycles in an associated envy graph. While this problem reduces to maximum-weight matching when starting from a clean slate (achieving polynomial-time solvability), we show that this is not the case when we account for preexisting conditions, such as residents' satisfaction with their original apartments. Whether the problem is polynomial-time solvable in the general case remains an intriguing open problem.

%	Our main contribution establishes that when all agents start from an equal baseline, finding the optimal allocation is equivalent to computing a maximum-value matching, yielding a polynomial-time algorithm. However, we pose a challenging open question: when agents have different starting positions (represented by a fixed directed graph $G_O$), can we still find the envy-minimizing allocation efficiently? Through concrete examples, we demonstrate that the optimal solution may no longer correspond to maximum-value matchings, suggesting that this generalization could be computationally harder.
	
%	This work bridges fair division theory with graph algorithms, offering both theoretical insights into cycle-based optimization and practical tools for equitable resource allocation in scenarios where historical context matters.
\end{abstract}

\section{Introduction}
The problem in this note is motivated by \emph{Fair Allocation of Improvements}.
In this setting there are $n$ agents, each of whom owns a single apartment in an old house. The house is demolished and a new house is constructed, wherein each agent  receives a new apartment.
Typically, each agent enjoys an \emph{improvement}, which can be measured by the difference between the new and the old apartment values. 
However, some agents might feel that their neighbors received a larger improvement, which might lead to envy. Our goal is to minimize that envy.%

In order to gain feedback from larger audiences, we have reduced our problem to a more general combinatorial problem: Min Max Average Cycle Weight, which may be interesting in its own right. We present it formally below.%
\footnote{
	We refer the interested reader to  \url{https://arxiv.org/abs/2504.16852} for more details on the fair improvement problem and how it is related to Min Max Average Cycle Weight.
%}
%\footnote{
	Envy minimization is also discussed by 
	\citet{nguyen2014minimizing,netzer2016distributed,peters2022robust}, using different techniques.
}

\section{Notation}
Let $N$ be a set of \emph{agents} and $M$ a set of \emph{objects}, with $|N|=|M|=n$.
Each agent $i\in N$ assigns a positive real value to each object $o\in M$, denoted $v_{i}(o)$.
We denote by $H_v$ the complete bipartite graph on the node set $N\sqcup M$, where each edge $i-o$ has value $v_{i,o}$.

An \emph{allocation} is a perfect matching in $H_v$, representing an assignment of exactly one object per agent.
We denote an allocation by $A$, where $A_i$ denotes the object assigned to agent $i\in N$.
We denote the set of all $n!$ possible allocations by $\allocs$.

Given an allocation $A\in \allocs$, the \emph{envy} of agent $i$ at agent $j$ 
%is the  difference between $j$'s value and $i$'s value, from the point of view of agent $i$, that is
is denoted $E_A(i,j)$ and defined as: $E_A(i,j) := v_i(A_j) - v_i(A_i)$.
The \emph{envy-graph} of assignment $A$, denoted $G_A$, is a complete directed graph with $N$ the set of nodes, and the weight of arc $i\to j$ is $E_A(i,j)$.%
\footnote{
We use the term ``value'' for the edges in the undirected bipartite graph $H_v$,
and the term ``weight'' for the arcs in the directed graph $G_A$.
}
\footnote{
The envy-graph has been introduced by \citet{aragones1995derivation} and used extensively in later works, e.g. \citet{halpern2019fair} and \citet{brustle2020one}.
}

We are interested in directed cycles in the graph $G_A$. Particularly, we are interested in the \emph{average weight} of each directed cycle (the sum of weights on the arcs in the cycle, divided by the number of arcs).

\begin{Definition}
For any directed graph $G$, the \emph{Maximum Average Cycle Weight}, denoted $MACW(G)$, is the 
maximum over all directed cycles in $G$, of the average cycle weight.
\end{Definition}
There are strongly-polynomial-time algorithms for computing $MACW(G)$ (refer to the Wikipedia page "Minimum mean weight cycle" for details).
%\citep{wikipedia-minimum-weight-cycle}

Every allocation $A$ yields a potentially different envy-graph $G_A$, with a potentially different maximum-average-cycle-weight $MACW(G_A)$. We are interested in finding an allocation for which $MACW(G_A)$ is as small as possible.

\begin{Example}
\label[example]{exm:macw}
There are 3 objects and 3 agents with valuations
    \[
\begin{bmatrix}
  & o_1 & o_2 &o_3 \\
  i_1 & 3 & 2 & 1\\
  i_2 & 3 & 5 & 7 \\
  i_3 & 7 & 8 & 9 \\
\end{bmatrix}.
\]
In the allocation $A_1 = \{o_1\}, A_2 = \{o_2\}, A_3 = \{o_3\}$, the weight of edge $i_1\to i_3$ is $-2$ and the weight of edge $i_3\to i_1$ is $-2$ too, so the weight of cycle $i_1\to i_3\to i_1$ is $-4$, and its average weight is $-2$.
Similarly, we can compute the average weight of each cycle in $G_A$. The results are summarized in 
\Cref{tab: cycle costs}.
%the table below.
The maximum average weights for each allocation are boldfaced. The allocation with smallest $MACW(G_A)$ is the second one, $(o_1, o_3, o_2)$, where the MACW is $-1/2$. Note that the total value of this allocation is $3+7+8=18$.
\begin{table}
\caption{
	\label{tab: cycle costs}
	Average cycle weights in \Cref{exm:macw}.
}
	\begin{center}
	\begin{tabular}{|c|c|c|c|c|c|c|c|}\toprule
	$A_1$ & $A_2$ & $A_3$& \makecell{Cycle \\$(i_1, i_2)$} & \makecell{Cycle \\$(i_1, i_3)$}& \makecell{Cycle \\$(i_2,i_3)$} &\makecell{Cycle \\$(i_1, i_2, i_3)$} & \makecell{Cycle \\$(i_1, i_3, i_2)$} \\
        \midrule
        $o_1$ &$o_2$ &$o_3$ & -3/2 & -4/2 & \textbf{1/2} & -1/3 & -5/3
        \\ 
        \midrule
        $o_1$ &$o_3$ &$o_2$ & -6/2 & -2/2 & \textbf{-1/2} & -5/3 & -4/3
        \\ 
        \midrule
        $o_2$ &$o_1$ &$o_3$ & \textbf{3/2} & -2/2 & 2/2 & 4/3 & -1/3
        \\ 
        \midrule
        $o_2$ &$o_3$ &$o_1$ & -3/2 & \textbf{2/2} & -2/2 & -4/3 & 1/3
        \\ 
        \midrule
        $o_3$ &$o_1$ &$o_2$ & \textbf{6/2} & 2/2 & 1/2 & 5/3 & 4/3
        \\
        \midrule
        $o_3$ &$o_2$ &$o_1$ & 3/2 & \textbf{4/2} & -1/2 & 1/3 & 5/3
        \\
		\bottomrule
	\end{tabular}
	\end{center}
\end{table}
\end{Example}

\section{Background results}
This section relates the average cycle weight to maximum-value matching. The results in this section are probably already known; we repeat them here as a background for the open question, presented in \Cref{sec:open}.

\begin{Proposition}
\label[proposition]{prop:maxsum}
An allocation $A\in \allocs$ satisfies $MACW(G_A)\leq 0$ if and only if $A$ is a maximum-value matching in $H_v$, that is, $\sum_{i\in N} v_i(A_i) \geq \sum_{i\in N} v_i(B_i)$ for all allocations $B\in \allocs$.
\end{Proposition}

In \Cref{exm:macw}, there is only one maximum-value allocation (the second one), and indeed its $MACW(G_A)$ is negative ($-1/2$).

\begin{proof}[Proof of \Cref{prop:maxsum}]
Suppose first that the $A$ is a maximum-value matching. Let $C$ be any directed cycle in $G_A$. W.l.o.g. denote $C = (1 \to 2 \to \cdots \to k \to 1)$. Then the total weight of $C$ is [where we denote $k+1 \equiv 1$]:
\begin{align*}
&
\sum_{i = 1}^k \left(v_i(A_{i+1}) - v_i(A_i)\right)
\\
=&
\left(\sum_{i = 1}^k v_i(A_{i+1})\right) -
\left(\sum_{i = 1}^k  v_i(A_i)\right)
\\
=&
\left(
\sum_{i = 1}^k v_i(A_{i+1})
+
\sum_{i = k+1}^n v_i(A_{i})
\right)
-
\left(
\sum_{i = 1}^k  v_i(A_i)
+
\sum_{i = k+1}^n v_i(A_{i})
\right)
\end{align*}
The right-hand term is the sum of agents' values in allocation $A$; the left-hand term is the sum of agents' values in an alternative allocation, in which $A_{i+1}$ is allocated to agent $i$ for all $i\in \{1,\ldots,k\}$ (and the other agents keep their items in $A$).
By assumption, $A$ maximizes the agents' values; hence, the sum is at most $0$.
This holds for all cycles; hence, $MACW(G_A)\leq 0$.

Suppose next that allocation $A$ is not a maximum-value matching, and let $B$ be a different allocation in which the sum of agents' values is higher.
We can transform $A$ into $B$ by switching objects between agents along one or more directed cycles. At least one of these cycles must strictly increase the sum of values of the involved agents. That is, there must be some cycle (denoted w.l.o.g. $C = (1 \to 2 \to \cdots \to k \to 1)$), such that  [where again we denote $k+1 \equiv 1$]: 
\begin{align*}
\sum_{i = 1}^k v_i(A_{i+1}) > \sum_{i = 1}^k v_i(A_{i}).
\end{align*}
The total weight of cycle $C$ is the difference between the LHS and the RHS, which is positive. Hence, its average weight is positive too, so  $MACW(G_A)>0$.
\end{proof}
From \Cref{prop:maxsum}, it is easy to deduce the following: 
\begin{Proposition}
\label[proposition]{prop:maxsum-macw}
An allocation $A\in \allocs$ minimizes $MACW(G_A)$ among all allocations, if and only if it is a maximum-value matching in $H_v$.
\end{Proposition}
\begin{proof}
Let allocation $A$ be a maximum-value matching, and let $B$ be any other allocation. 
We have to prove that $MACW(A)\leq MACW(B)$. Consider two cases:

\underline{Case 1:} 
$B$ is not a maximum-value matching.
Then  \Cref{prop:maxsum} implies that $MACW(G_B)> 0 \geq MACW(G_A)$, and we are done.

\underline{Case 2:} 
$B$ is a maximum-value matching.
Then  \Cref{prop:maxsum} implies that $MACW(G_B)\leq 0$ and  $MACW(G_A)\leq 0$.
We can transform $A$ into $B$ by switching objects along one or more directed cycles. The weight of each of these cycles is at most $0$; but since the sum of agents' values in $A$ is the same as in $B$, The weight of each of these cycles must be exactly $0$. Hence, $MACW(G_B) = 0 = MACW(G_A)$, and we are done.
\end{proof}

As a maximum-weight matching can be found in polynomial time, we get:
\begin{Corollary}
\label[corollary]{cor:macw}
It is possible to find in polynomial time, an allocation $A\in\allocs$ for which $MACW(G_A)$ is minimized.
\end{Corollary}

\section{Open Question}
\label{sec:open}
Now, we would like to extend \Cref{cor:macw} in the following way.

We are given a fixed complete directed graph $G_O$ with node-set $N$.%
\footnote{In our motivating application, $G_O$ represents the envy caused by the original apartments.}
Each edge $i\to j$ in $G_O$ has some weight $E_O(i,j)$.
Denote by $G_A-G_O$ the ``arcwise-difference'' of the two graphs, that it, a complete directed graph in which the weight of edge $i\to j$ is the difference of the weights of the corresponding edges in $G_A$ and $G_O$: $E_A(i,j) - E_O(i,j)$.%
\footnote{
In our motivating application, $G_A-G_O$ represents the envy caused by the improvements.
}

We now present the Min Max Average Cycle Weight problem.
\begin{Open}
Given any fixed graph $G_O$, is it possible to find in polynomial time, an allocation $A\in\allocs$ for which $MACW(G_A-G_O)$ is minimized?
\end{Open}
When $G_O$ assigns zero weights to all arcs, \Cref{cor:macw} implies that the answer is ``yes'', as we can simply compute a maximum-value matching in $H_v$.
But when $G_O$ has even a single arc with a non-zero weight, the solution may no longer be a maximum-value matching. 

\begin{Example}
\label[example]{exm:macw-2}
Continuing \Cref{exm:macw}, suppose $E_O(i_1,i_3) = +4$  and all other weights in $G_O$ are $0$. Then 
the average cycle weights in $G_A-G_O$ are as in 
\Cref{tab: cycle costs-2}.
%the table below.
\begin{table}
\caption{
	\label{tab: cycle costs-2}
Average cycle weights in \Cref{exm:macw-2}.
}
	\begin{center}
	\begin{tabular}{|c|c|c|c|c|c|c|c|}\toprule
	$A_1$ & $A_2$ & $A_3$& \makecell{Cycle \\$(i_1, i_2)$} & \makecell{Cycle \\$(i_1, i_3)$}& \makecell{Cycle \\$(i_2,i_3)$} &\makecell{Cycle \\$(i_1, i_2, i_3)$} & \makecell{Cycle \\$(i_1, i_3, i_2)$} \\
        \midrule
        $o_1$ &$o_2$ &$o_3$ & -3/2 & -8/2 & \textbf{1/2} & -1/3 & -9/3
        \\ 
        \midrule
        $o_1$ &$o_3$ &$o_2$ & -6/2 & -6/2 & \textbf{-1/2} & -5/3 & -8/3
        \\ 
        \midrule
        $o_2$ &$o_1$ &$o_3$ & \textbf{3/2} & -6/2 & 2/2 & 4/3 & -5/3
        \\ 
        \midrule
        $o_2$ &$o_3$ &$o_1$ & -3/2 & \textbf{-2/2} & \textbf{-2/2} & -4/3 & \textbf{-3/3}
        \\ 
        \midrule
        $o_3$ &$o_1$ &$o_2$ & \textbf{6/2} & -2/2 & 1/2 & 5/3 & 0/3
        \\
        \midrule
        $o_3$ &$o_2$ &$o_1$ & \textbf{3/2} & 0/2 & -1/2 & 1/3 & 1/3
        \\
		\bottomrule
	\end{tabular}
	\end{center}
\end{table}

The maximum average weights for each allocation are boldfaced. The allocation with smallest $MACW(G_A)$ is the fourth one, $(o_2, o_3, o_1)$,
with total value $2+7+7=16$; it is different than the maximum-value allocation.
\end{Example}

\newpage
\section*{Acknowledgments}
	The examples were corrected and verified by Aristotle.harmonic.fun \citep{achim2025aristotleimolevelautomatedtheorem}.

\bibliographystyle{ACM-Reference-Format}
\bibliography{ref}

\newpage
\end{document}